\documentclass[runningheads]{llncs}
\usepackage[T1]{fontenc}
\usepackage{amsfonts, amsmath,amssymb}
\usepackage{hyperref,graphicx, enumitem}
\usepackage{color}

\begin{document}

\title{The Sources of Unknowability and Self-refutation in Epistemic and Dynamic Epistemic Logic}
\titlerunning{The Sources of Unknowability and Self-refutation}
\author{Eiji Yamada\inst{1}\orcidID{0009-0007-2519-8016}}

\authorrunning{E. Yamada}

\institute{University of Tsukuba, Tsukuba, Ibaraki 305-8571, Japan\\
\email{s2420142@u.tsukuba.ac.jp}}

\maketitle
\begin{abstract}
In this paper, we define a formula $\varphi$ to be unbelievable if $\Box\varphi$ is unsatisfiable, and unknowable if $\Box\varphi\land\varphi$ is unsatisfiable. We then analyze the sources of unknowability and unbelievability in different classes of frames K, KD, KD45, and S5. We first show that any unbelievable formula is a fixed point of the Moore function defined by $f(\varphi)=\varphi\land\lnot\Box\varphi.$ Our main result shows that in S5, the static notions of unknowability and unbelievability in epistemic logic, and the dynamic notions of always informativeness and eventual self-refutation in dynamic epistemic logic are all equivalent to ``Moorean phenomena.'' We also generalize the result to the multi-agent case, showing that although all unknowable formulas still manifest Moorean phenomena, new mechanisms arise due to their interactive nature. Finally, we briefly analyze whether the Brandenburger-Keisler paradox in epistemic game theory can be considered Moorean.  
\end{abstract}
\section{Introduction}
In epistemic logic, one common interpretation of the formula $\Box\varphi$ is ``an agent believes $\varphi$'' and that of $\Box\varphi\land\varphi$ is ``an agent knows $\varphi$''. Then, a natural question arises, ``is there a formula that can be true but cannot be known by some or any agents?'' Although the discussions of unknowability have been made by philosophers over the years, there seem to be few papers that formally discuss it using epistemic logic. In this paper, we define the notions of unknowability and unbelievability and analyze them in the class of frames K, KD, KD45, and S5. We define a formula $\varphi$ to be unbelievable if $\Box\varphi$ is unsatisfiable, and to be unknowable if $\Box\varphi\land\varphi$ is unsatisfiable. This definition is a little different from those in the existing literature but it is the simplest, most straightforward, and static one, allowing us to systematically characterize unknowable and unbelievable formulas.

On the other hand, Holliday \& Icard~\cite{holliday2010} analyzes unsuccessful, self-refuting, informative, learnable, and non Cartesian formulas in the context of public announcement and proves connections with ``Moorean phenomena'' in the single-agent case. These notions except for non-Cartesianness are rather dynamic ones. (\cite{holliday2010} defines $\varphi$ to be Cartesian if $\Box\varphi$ is satisfiable, so non-Cartesianess is the same as unbelievability in our sense.) The sentence $p\land\lnot\Box p$ (``$p$ is true but I do not believe $p$'') is called the \emph{Moore sentence} and the absurdity of asserting it has been discussed by philosophers as \emph{Moore's paradox} (\cite{hintikka2005,moore1942}). 

Our main result (Theorem \ref{Unknowability and unbelievability in KD45 and S5}) connects the static notions of unbelievability and unknowability in epistemic logic and the dynamic notions of always informativeness and eventual self-refutation in dynamic epistemic logic with two different views of Moorean phenomena. Although the equivalence between non-Cartesianness, always informativeness and eventual self-refutation is due to Holliday \& Icard~\cite{holliday2010}, our theorem fills the gap between their result and our notions of unknowability and unbelievability as well as two different views of Moorean phenomena: viewed either as a fixed-point of the Moore function, or a contradiction among propositional formulas and formulas that are believed/considered possible. The result is not only philosophically and linguistically interesting but also expected to give theoretical foundations for multi-agent systems, security protocols and AI: where information sharing is restricted, identifying the sources helps unravel epistemic complexities. For the relationship between epistemic logic and these fields, see~\cite{fagin1995,meyer1995}. We also extend the theorem to the multi-agent version for the static notions (Theorem~\ref{Multi-agent unknowability}).

Finally, we briefly introduce the Brandenburger-Keisler paradox in the epistemic game theory literature~\cite{abramsky2012,brandenburger2006,pacuit2007}. The Brandenburger-Keisler paradox is a contradiction that arises from considering the sentence ``Ann believes that Bob assumes that Ann believes that Bob's assumption is wrong.'' Formally, this means that the formula $\Box_a\boxplus_b D$ is unsatisfiable in KD (even if $\boxplus_b D$ is satisfiable) where $\boxplus_b$ is called the assumption operator and $D$ is a proposition letter which means ``Ann believes that Bob's assumption is wrong.'' We conclude that although the paradox is still considered a Moorean phenomenon, the formula $\boxplus_b D$ is already unbelievable to Ann in KD due to Russell-style self-reference. This is in contrast to the classic Moore sentence $p\land\lnot\Box p$, which requires not only belief consistency but also positive introspection to be unbelievable.
\section{Unknowable and unbelievable formulas in the single agent case}
In this section, we consider unknowable and unbelievable formulas in the single agent case. Hereafter, let $\mathbf{Prop}$ be a set of proposition letters.
\begin{definition}
    Define \emph{formulas} in epistemic logic by
    \begin{equation*}
        \varphi:=p\mid \lnot\varphi\mid\varphi\land\psi\mid\Box\varphi\quad(p\in\mathbf{Prop})
    \end{equation*}
\end{definition}
Other propositional connectives and the $\Diamond$ operator are defined as abbreviations. (Kripke) models $M=(W,R,V)$, frames $\mathcal{F}=(W,R)$, as well as satisfiability of a formula in a class of frames is defined as usual.
\begin{definition}
Let $\varphi$ be a formula, $M$ be a model, and $\mathcal{K}$ be a class of frames.
\begin{itemize}
    \item $\varphi$ is \emph{believed} at $w$ in $M$ iff $M,w\models\Box\varphi$
    \item $\varphi$ is \emph{believable} in $\mathcal{K}$ iff $\Box\varphi$ is satisfiable in $\mathcal{K}$
    \item $\varphi$ is \emph{known} at $w$ in $M$ iff $M,w\models\varphi\land\Box\varphi$
    \item $\varphi$ is \emph{knowable} in $\mathcal{K}$ iff $\varphi\land\Box\varphi$ is satisfiable in $\mathcal{K}$

\end{itemize}
``\emph{Unbelievable}'' means ``not believable'' and ``\emph{unknowable}'' means ``not knowable.''
\end{definition}
\begin{remark}
    Although at least KD4 is often assumed for belief, we use the word ``believe'' for $\Box\varphi$ in any class of frames, based on the interpretation that an agent believes $\varphi$ if $\varphi$ holds at any state he considers possible. This convention allows us to see the roles of belief consistency and introspections, hence our choice K, KD, KD45, and S5 (it is possible to consider other classes such as S4 (=KT4) but we omit it here since the results are expected to be similar to those for KD). By definition, ``$\varphi$ is unbelievable in $\mathcal{K}$'' means that an agent cannot believe $\varphi$ at any state in any model in $\mathcal{K}$. This means that there is no situation in which an agent believes $\varphi$, due to the axioms that characterize the class of frames $\mathcal{K}$ (e.g., belief consistency, positive or negative introspection). The term ``unbelievable'' refers to the global, logical impossibility of believing a proposition and is unrelated to the everyday sense in which something is merely psychologically hard to believe. The same remark applies to unknowability.  
    
    Traditionally, $\Diamond K\varphi$ is used in the philosophical literature to express knowability where $\Diamond$ means the more abstract ``possibility'' and $K$ means knowledge or belief. \cite{balbiani2008} defines $\varphi$ to be knowable if there exists a formula $\psi$ such that after announcing $\psi$, $K\varphi$ holds. (\cite{holliday2010} suggests the word \emph{learnable} for this notion. See Definition~\ref{learnable}.) Thus, our definition is rather simpler, static one. Satisfiability in our definition of unknowability partially corresponds to the possibility represented by the $\Diamond$. Note also that $\varphi$ is unknowable iff $\varphi\to\lnot\Box\varphi$ is valid. 
\end{remark}
The following lemma states that in KD45 and S5, the truth of ``modal atoms'' do not change before and after moving between two states.
\begin{lemma}\label{modal agreement lemma}
    Let $M=(W,R,V)$ be a KD45 (or S5) model. If $wRv$, then for any formula $\varphi$ of the form $\Box\psi$ or $\Diamond\psi$, we have 
    \begin{equation*}
        M,w\models\varphi\Longleftrightarrow M,v\models\varphi
    \end{equation*}
\end{lemma}
\begin{proof}
    For $\varphi=\Box\psi$, $(\Rightarrow)$ follows from transitivity and $(\Leftarrow)$ follows from Euclideanness. For $\varphi=\Diamond\psi$, the order of the properties is reversed. 
\end{proof}
\begin{lemma}\label{fixed point}
In K, KD, KD45 and S5, $\varphi$ is unknowable iff $\varphi\leftrightarrow(\varphi\land\lnot\Box\varphi)$ is valid. In K and KD, if $\varphi$ is unbelievable, then $\varphi$ is unknowable but the converse does not hold. In KD45 and S5, $\varphi$ is unbelievable iff $\varphi$ is unknowable.
\end{lemma}
\begin{proof}
    For the first statement, $\varphi$ is unknowable $\Longleftrightarrow$ $\varphi\land\Box\varphi$ is unsatisfiable $\Longleftrightarrow$ $\varphi\to\lnot\Box\varphi$ is valid $\Longleftrightarrow$ $\varphi\leftrightarrow (\varphi\land\lnot\Box\varphi)$ is valid. For the second statement, unbelievability clearly implies unknowability. For the converse direction, the Moore sentence $\varphi:=p\land\lnot\Box p$ is unknowable but believable in K and KD. In fact, for the KD-model $M=(W,R,V)$ defined by $W=\{w_1,w_2\}$, $R=\{(w_1,w_2), (w_2,w_1)\}$ and $V(p)=\{w_2\}$, we have $M,w_1\models\Box\varphi$, so $\varphi$ is believable. However,
        $\Box\varphi\land\varphi\equiv\Box(p\land\lnot\Box p)\land (p\land\lnot\Box p)\equiv\Box p\land\Box\lnot\Box p\land p\land\lnot\Box p\equiv\bot$, so $\varphi$ is unknowable. For the third statement, unbelievability again implies unknowability. We prove that unknowability implies unbelievability by contraposition. If $\varphi$ is believable, then $M,w\models\Box\varphi$ for some KD45 (or S5) pointed model $M,w$. By seriality, there is a $v$ such that $wRv$. Then, $M,v\models\varphi$ and by lemma \ref{modal agreement lemma}, we have $M,v\models\Box\varphi$, so $\Box\varphi\land\varphi$ is satisfiable. Thus, $\varphi$ is knowable.
\end{proof}
So, every unbelievable formula $\varphi$ is equivalent to the Moore sentence $\varphi\land\lnot\Box\varphi$ and is the component of that very sentence. In other words, every unbelievable formula is a fixed point of the \emph{Moore function} $f(\varphi):=\varphi\land\lnot\Box\varphi$.\footnote{Fixed points frequently appear anywhere self-reference is involved.} Furthermore, unknowability and being equivalent to a Moore sentence are the same.

Note that the Moore sentence $p\land\lnot\Box p$ needs only KD to be unknowable but needs KD4 to be unbelievable. In fact, it is believable in KD and KD5: for KD5, consider the KD5 model $M=(W,R,V)$ with $W=\{w_1,w_2,w_3\}$,
$R=\{(w_1,w_2), (w_2,w_2), (w_2,w_3),(w_3,w_2),(w_3,w_3)\}$ and $V(p)=\{w_2\}$. Then, $M,w_1\models\Box(p\land\lnot\Box p)$. In general, the equivalence between unknowability and unbelievability holds in any class of frames contained in KD4 as we can see from the above proof. This means that the two notions are the same under belief consistency and positive introspection.
\subsection{Unknowability and unbelievability in K and KD}
\begin{theorem}[Unbelievability in K and KD]
\begin{enumerate}
    \item Every formula is believable in K.
    \item  A formula $\varphi$ is unbelievable in KD iff $\varphi$ is logically equivalent to $\bot$ in KD (i.e., $\varphi$ is unsatisfiable in KD).
\end{enumerate}
\end{theorem}
\begin{proof}
    (1) Consider $M=(W,R,V)$ with $W=\{w\}$ and $R=\varnothing$. (2) $(\Leftarrow)$ is trivial. $(\Rightarrow)$ We prove the contraposition. If $\varphi$ is satisfiable in KD, $M,w\models\varphi$ for some KD model $M=(W, R, V)$. Let $M'=(W', R', V')$ be the KD model such that $W'=W\cup \{w'\}$, $R'=R\cup \{(w',w)\}$ and $V'$ is an extension of $V$ where $w'$ is a new state. Then, $M',w'\models\Box\varphi$ so $\Box\varphi$ is satisfiable in KD.
\end{proof}
Every formula is logically equivalent to a formula in disjunctive normal form in K as defined in Definition~\ref{DNF in K} (see for example, \cite{fine1975}). The notations in Definition~\ref{DNF notation} are from~\cite{holliday2010}.
\begin{definition}\label{DNF in K}
    A formula is in \emph{disjunctive normal form} (\emph{DNF}) in K iff it is a disjunction of conjunctions of the form $\delta\equiv\alpha\land\Box\beta_1\land\cdots\land\Box\beta_n\land\Diamond\gamma_1\land\cdots\land\Diamond\gamma_m$ where $\alpha$ is a conjunction of literals and $\beta_i$ and $\gamma_i$ are arbitrary formulas.
\end{definition}
\begin{definition}\label{DNF notation}
    Given $\delta\equiv\alpha\land\Box\beta_1\land\cdots\land\Box\beta_n\land\Diamond\gamma_1\land\cdots\land\Diamond\gamma_m$ in disjunctive normal form, we define $\delta^\alpha\equiv\alpha$, $\delta^{\alpha\Box}\equiv\alpha\land\Box\beta_1\land\cdots\land\Box\beta_n$, and similarly for $\delta^{\alpha\Diamond}$, $\delta^\Box$, $\delta^{\Box\Diamond}$, and $\delta^{\Diamond}$.
\end{definition}
\begin{theorem}[Unknowability in K and KD]\label{Unknowability in K and KD}
    Let $\bigvee_{k\in K}\delta_k$ be a disjunctive normal form of a formula $\varphi$ in K. Let $B_k$, $\Gamma_k$ be the sets of formulas $\beta$, $\gamma$ such that $\Box\beta$ is in $\delta_k^\Box$ and $\Gamma_k$ is in $\delta_k^\Diamond$, respectively.\footnote{That is, for $\delta_k=\alpha_k\land\Box\beta_{k1}\land\cdots\Box\beta_{kn}\land\Diamond\gamma_{k1}\land\cdots\land\Diamond\gamma_{km}$, let $B_k=\{\beta_{kj}\colon j=1,\ldots, n\}$ and $\Gamma_k=\{\gamma_{kj}\colon j=1,\ldots, m\}$.}
    \begin{enumerate}
        \item $\varphi$ is unknowable in K iff for all $k\in K$, either (i) $\delta_k^\alpha$ is unsatisfiable in K, or (ii) there exists a $\gamma\in \Gamma_k$ such that $B_k\cup\{\varphi\}\cup\{\gamma\}$ is unsatisfiable.
        \item $\varphi$ is unknowable in KD iff  for all $k\in K$, one of the following conditions hold: (i) $\delta_k^\alpha$ is unsatisfiable in KD, (ii) $B_k\cup \{\varphi\}$ is unsatisfiable in KD, or (iii) there exists a $\gamma\in \Gamma_k$ such that $B_k\cup\{\varphi\}\cup\{\gamma\}$ is unsatisfiable.
    \end{enumerate}
\end{theorem}
\begin{proof}
    We only sketch the proof here. First note that $\varphi\land\Box\varphi\equiv(\bigvee_{k\in K}\delta_k)\land\Box\varphi\equiv\bigvee_{k\in K}(\delta_k\land\Box\varphi)\equiv\bigvee_{k\in K}[\delta_k^\alpha\land(\delta_k^\Box\land\Box \varphi)\land\delta_k^\Diamond]$. In general, we can prove the following lemmas for any propositional formula $\alpha$ and any finite sets of formulas $B$ and $\Gamma$: In K, $\alpha\land\Box(\bigwedge B)\land\bigwedge_{\gamma\in\Gamma}\Diamond\gamma$ is satisfiable iff $\alpha$ is satisfiable, and for every $\gamma\in\Gamma$, $B\cup\{\gamma\}$ is satisfiable. In KD, we further require that $B$ is satisfiable. The theorem then immediately follows from these lemmas.
\end{proof}
$B_k$ is the set of formulas that an agent originally believed and $\Gamma_k$ is the set of formulas that an agent originally considered possible. Thus, the unsatisfiability of $B_k\cup\{\varphi\}\cup\{\gamma\}$ in the above theorem says that if $\varphi$ is unknowable, then knowing $\varphi$ would lead to a contradiction among propositions that an agent originally believed or considered possible, and the proposition he is now trying to believe. Recall that the Moore sentence $\varphi:=p\land\lnot\Box p$ is believable in KD but unknowable in KD. This means that in lack of introspections, it is possible that an agent consistently believes $p$ while believing $\Diamond\lnot p$ when $\varphi$ is actually false. However, it is impossible to believe $\varphi$ when $\varphi$ is true, as indicated by the unsatisfiability of $B_k\cup\{\varphi\}\cup\{\gamma\}$. Therefore, this generalized view in Theorem~\ref{Unknowability in K and KD} is considered the best way to capture ``Moorean phenomena.''
\subsection{Unknowability and unbelievability in KD45 and S5}
In KD45 and S5, we can use the strong fact that every formula can be converted into a disjunctive normal form with modal depth at most one. \cite{holliday2010} defines DNF in K45 and states that every formula is logically equivalent to a DNF in K45.
\begin{definition}\label{Def of DNF}
    A formula is in \emph{disjunctive normal form} (\emph{DNF}) in K45 iff it is a disjunction of conjunctions of the form $\delta\equiv\alpha\land\Box\beta_1\land\cdots\land\Box\beta_n\land\Diamond\gamma_1\land\cdots\land\Diamond\gamma_m$ where $\alpha$ and $\gamma_i$ are conjunctions of literals and each $\beta_i$ is a disjunction of literals.
\end{definition}
For convenience, we state the reduction rules for K45, which are used when converting a formula into DNF in K45 (see \cite{georgiev2017}). 
\begin{lemma}\label{reduction rules for K45}
   In K45, the following reduction rules hold.
\[
\begin{aligned}
1.\;& \Box\Diamond A \equiv (\Diamond A \vee \Box \bot)
    && 
    & 2.\;& \Diamond\Box A \equiv (\Box A \wedge \Diamond \top)
\\
3.\;& \Box\Box A \equiv \Box A
    &&
    & 4.\;& \Diamond\Diamond A \equiv \Diamond A
\\
5.\;& \Diamond(\Diamond A \wedge B) \equiv (\Diamond A \wedge \Diamond B)
    &&
    & 6.\;& \Diamond(\Box A \wedge B) \equiv (\Box A \wedge \Diamond B)
\\
7.\;& \Box(\Diamond A \vee B) \equiv (\Diamond A \vee \Box B)
    &&
    & 8.\;& \Box(\Box A \vee B) \equiv (\Box A \vee \Box B)
\end{aligned}
\]
\end{lemma}
We prepare two lemmas (Lemmas \ref{Normal form of Boxvarphi in KD45} and \ref{satisfiability and modal depth zero}) for Theorem~\ref{Unknowability and unbelievability in KD45 and S5}.
\begin{lemma}\label{Normal form of Boxvarphi in KD45}
    If $\bigvee_{k\in K}\delta_k$ is a disjunctive normal form of $\varphi$ in K45, then in KD45 and S5, $\Box\varphi$ is logically equivalent to
    \begin{equation*}
        \bigvee_{S\subseteq K}\left[\bigwedge_{k\in S}\delta_k^{\Box\Diamond}\land\Box(\bigvee_{k\in S}\delta_k^{\alpha})\right]
    \end{equation*}
\end{lemma}

\begin{proof}
    We show that $\Box(\bigvee_{k\in K}(\delta_k^{\alpha}\land\delta_k^{\Box\Diamond}))\leftrightarrow \bigvee_{S\subseteq K}\left[\bigwedge_{k\in S}\delta_k^{\Box\Diamond}\land\Box(\bigvee_{k\in S}\delta_k^{\alpha})\right]$ is valid in KD45. Let $M$ be a KD45 model and $w$ be a state in $M$. \\$(\Rightarrow)$ Suppose that $M,w\models\Box(\bigvee_{k\in K}(\delta_k^{\alpha}\land\delta_k^{\Box\Diamond}))$. Let $S=\{k\in K\colon M,w\models \delta_k^{\Box\Diamond}\}$. Then, clearly $M,w\models\bigwedge_{k\in S}\delta_k^{\Box\Diamond}$. To show $M,w\models \Box\bigvee_{k\in S}\alpha_k$, take any $v$ such that $wRv$. Then $M,v\models \bigvee_{k\in K}(\delta_k^{\alpha}\land \delta_k^{\Box\Diamond})$ by assumption, so $M,v\models\delta_k^{\alpha}\land \delta_k^{\Box\Diamond}$ for some $k\in K$. Since $M,v\models \delta_k^{\Box\Diamond}$ and $\delta_k^{\Box\Diamond}$ is purely modal, $M,w\models \delta_k^{\Box\Diamond}$ holds by lemma \ref{modal agreement lemma} so that $k$ must be in $S$. Thus, we have $M,v\models\bigvee_{k\in S}\delta_k^{\alpha}$ hence $M,w\models \Box\bigvee_{k\in S}\delta_k^\alpha$, which implies that $M,w\models\left[\bigwedge_{k\in S}\delta_k^{\Box\Diamond}\land\Box(\bigvee_{k\in S}\delta_k^{\alpha})\right]$.\\
    $(\Leftarrow)$ Suppose that $M,w\models \bigvee_{S\subseteq K}\left[\bigwedge_{k\in S}\delta_k^{\Box\Diamond}\land\Box(\bigvee_{k\in S}\delta_k^{\alpha})\right]$. Then $M,w\models \left[\bigwedge_{k\in S}\delta_k^{\Box\Diamond}\land\Box(\bigvee_{k\in S}\delta_k^{\alpha})\right]$ for some $S\subseteq K$, and in particular, $M,w\models\Box(\bigvee_{k\in S}\delta_k^\alpha)$. To show $M,w\models\Box(\bigvee_{k\in K}(\delta_k^{\alpha}\land\delta_k^{\Box\Diamond}))$, take any $v$ such that $wRv$. We then have $M,v\models(\bigvee_{k\in S}\delta_k^\alpha)$, and by Lemma \ref{modal agreement lemma}, $M,v\models \left[\bigwedge_{k\in S}\delta_k^{\Box\Diamond}\land\Box(\bigvee_{k\in S}\delta_k^{\alpha})\right]$
    . Thus, $M,v\models\delta_k^\alpha\land\delta_k^{\Box\Diamond}$ for some $k\in S$. Thus, we have $M,w\models\Box(\bigvee_{k\in K}(\delta_k^{\alpha}\land\delta_k^{\Box\Diamond}))$.
\end{proof}
\begin{lemma}\label{satisfiability and modal depth zero}
    Let $\beta$ be a propositional formula and $\Gamma$ be a finite set of propositional formulas. Then, in KD45 and S5, $\Box\beta\land \bigwedge_{\gamma\in\Gamma}\Diamond\gamma$ is satisfiable iff both of the following conditions hold:
    \begin{enumerate}
        \item $\beta$ is satisfiable.
        \item For every $\gamma\in\Gamma$, $\beta\land\gamma$ is satisfiable.
    \end{enumerate}
\end{lemma}
\begin{proof}
        ($\Rightarrow$) Assume $M,w\models \Box\beta\land \bigwedge_{\gamma\in\Gamma}\Diamond\gamma$ for some KD45 (or S5) model $M=(W,R,V)$ and $w\in W$. For (1), we have $M,v\models\beta$ for some $v$ with $wRv$ by seriality, so $\beta$ is satisfiable in KD45 (or S5). For (2), for every $\gamma\in\Gamma$, there is a $v_\gamma$ with $wRv_\gamma$ such that $M,v_\gamma\models\gamma$ by assumption. Thus, we have $M,v_\gamma\models\beta\land \gamma$ again by assumption. Therefore, $\beta\land\gamma$ is satisfiable in KD45 (or S5). ($\Leftarrow$) Assume (1) and (2). Then, since $\beta$ and $\gamma$ are propositional, there are valuation maps $V_\beta$, $V_{\beta\land\gamma}$ and states $v_\beta$, $v_{\beta\land\gamma}$ such that $V_\beta$ makes $\beta$ true at $v_\beta$ and $V_{\beta\land\gamma}$ makes $\beta\land\gamma$ true at $v_{\beta\land\gamma}$. Define the S5 model $M=(W,R,V)$ by letting $W=\{v_\beta\}\cup\{v_{\beta\land\gamma}\colon \gamma\in\Gamma\}$, $R=W\times W$, and $V$ be defined in the obvious manner. This $M$ satisfies $M,v_\beta\models\Box\beta\land\bigwedge_{\gamma\in\Gamma}\Diamond\gamma$, hence $\Box\beta\land\bigwedge_{\gamma\in\Gamma}\Diamond\gamma$ is satisfiable.
\end{proof}
Next, we define several notions in public announcement. For a general reference on public announcement and dynamic epistemic logic, see~\cite{vanditmarsch2008}. Definitions~\ref{relativized model}-~\ref{learnable} are from \cite{holliday2010}. 
\begin{definition}\label{relativized model}
    Let $M=(W, R, V)$ be a model and $\varphi$ be a formula. The \emph{relativization} of $M$ to $\varphi$ is the submodel $M_{|\varphi}=(W_{|\varphi}, R_{|\varphi}, V_{|\varphi})$ where $W_{|\varphi}=\{w\in W\colon M,w\models\varphi\}$, $R_{|\varphi}=R\cap W_{|\varphi}$, and $V_{|\varphi}(p)=V(p)\cap W_{|\varphi}$. 
\end{definition}
Informally, $\varphi$ is successful if $\varphi$ remains true after $\varphi$ is announced, and $\varphi$ is self-refuting if $\varphi$ becomes false after $\varphi$ is announced. 
\begin{definition}\label{successful and self-refuting}
    A formula $\varphi$ is \emph{successful} iff for every pointed model, $M,w\models\varphi\land\Diamond\varphi$ implies $M_{|\varphi},w\models\varphi$. A formula $\varphi$ is \emph{unsuccessful} iff it is not successful. A formula is \emph{self-refuting} iff for every pointed model, $M,w\models\varphi\land \Diamond\varphi$ implies $M_{|\varphi} ,w\not\models\varphi$.
\end{definition}
The precondition $M,w\models\varphi\land\Diamond\varphi$ is to ensure that $\varphi$ is true not only at $w$ but also at some accessible state so that $M_{|\varphi}$ remains KD45 provided $M$ is. However, we can replace it with $M,w\models\varphi$ since we consider only S5 for dynamic notions in this paper.

Informally, $\varphi$ is informative if an agent's accessibility relation changes after $\varphi$ is announced. $\varphi$ is eventually self-refuting if $\varphi$ becomes false after $\varphi$ is announced finitely many times.
\begin{definition}\label{informative}
    A formula $\varphi$ is \emph{informative} iff there is a pointed model such that $M,w\models\varphi$ and $M_{|\varphi}\neq M$. Otherwise $\varphi$ is \emph{uninformative}. A formula $\varphi$ is \emph{always informative} iff for all pointed models such that $M,w\models\varphi$, $M_{|\varphi}\neq M$.
\end{definition}
\begin{definition}\label{eventually self-refuting}
    Given a model $M$, we define $M_{|^n\varphi}$ recursively by $M_{|^0\varphi}=M$, $M_{|^{n+1}\varphi}=(M_{|^n\varphi})_{|\varphi}$. A formula $\varphi$ is \emph{self-refuting within} $n$ \emph{steps} iff for all pointed models, if $M,w\models\varphi$, then $M_{|^m\varphi},w\not\models\varphi$ for some $m\leq n$; $\varphi$ is \emph{eventually self-refuting} iff for all pointed models, if $M,w\models\varphi$, then there is an $n$ such that $M_{|^n\varphi},w\not\models\varphi$.
\end{definition}
We state the definition of learnability for comparison. 
\begin{definition}\label{learnable}
    A formula $\varphi$ is \emph{learnable} iff for all pointed models, if $M,w\models\varphi$, then there is some $\psi$ such that $M_{|\psi},w\models\Box\varphi$.
\end{definition}
Finally, we state our main result, which connects the static notions of unbelievability and unknowability in epistemic logic, and the dynamic notions of always informativeness and eventual self-refutation in dynamic epistemic logic.
\begin{theorem}[Unknowability and unbelievability in KD45 and S5]\label{Unknowability and unbelievability in KD45 and S5}
    Let $\bigvee_{k\in K}\delta_k$ be a disjunctive normal form of $\varphi$ in K45, and for each $S\subseteq K$, let $B_S$, $\Gamma_S$ be the sets of formulas $\beta$, $\gamma$ such that $\Box\beta$, $\Diamond\gamma$ are in the formula $\bigwedge_{k\in S}\delta_k^{\Box\Diamond}$, respectively. Then in KD45 and S5, the following are equivalent:
    \begin{enumerate}
        \item $\varphi$ is unbelievable.
        \item $\varphi$ is unknowable.
        \item $\varphi\leftrightarrow(\varphi\land\lnot\Box\varphi)$ is valid (i.e., $\varphi$ is a fixed point of the Moore function).
        \item For any $S\subseteq K$, (i) $B_S\cup \{\bigvee_{k\in S}\delta_k^{\alpha}\}$ is unsatisfiable or (ii) there exists a $\gamma\in \Gamma_S$ such that $B_S\cup\{\bigvee_{k\in S}\delta_k^{\alpha}\}\cup\{\gamma\}$ is unsatisfiable.
        \end{enumerate}
        In S5, the following condition is also equivalent to the above conditions.
        \begin{enumerate}
        \item[5.] $\varphi$ is always informative.
        \item[6.] $\varphi$ is eventually self-refuting.
        \end{enumerate}
\end{theorem}
\begin{proof}
    $(1)\Leftrightarrow (2)\Leftrightarrow (3)$ is Lemma~\ref{fixed point}. $(1)\Leftrightarrow (4)$ follows from Lemmas~\ref{Normal form of Boxvarphi in KD45} and \ref{satisfiability and modal depth zero}. $(1) \Leftrightarrow (5) \Leftrightarrow (6)$ in S5 is exactly Proposition 5.7 in \cite{holliday2010}.\footnote{\cite{holliday2010} used the term not Cartesian in Proposition 5.7 and found the term unknowable ``seems more natural'' for that notion. The definitions of unknowability and unbelievability, as well as their connection to Moorean phenomena presented in this paper, were developed independently before the author encountered Holliday \& Icard~\cite{holliday2010}.}
\end{proof}
\begin{remark}
We briefly remark on the results in other classes of frames. First, as we mentioned after the proof of Lemma~\ref{fixed point}, $(1)\Leftrightarrow (2)\Leftrightarrow (3)$ holds in any class of frames contained in KD4. Also, $(1)\Rightarrow (5)$ holds in any class of frames: In fact, suppose that $\varphi$ is unbelievable and take any model $M,w$. Since $\Box\varphi$ is unsatisfiable, $M,w\not\models\Box\varphi$ so that there is a $v\in M$ such that $wRv$ and $M,v\models\lnot\varphi$. Then, $v\in M$ but $v\notin M_{|\varphi}$ so $M_{|\varphi}\neq M$ hence $\varphi$ is always informative. Furthermore, $(6)\Rightarrow (5)$ holds in any class of frames: In fact, to show the contraposition, suppose that $\varphi$ is not always informative. Then, there is a model $M,w$ such that $M,w\models\varphi$ and $M_{|\varphi}=M$, so $\varphi$ cannot be eventually self-refuting. Also, $(5)\Rightarrow (1)$ actually holds in KD45 as well since the proof of this direction in Proposition 5.7 in Holliday \& Icard~\cite{holliday2010} uses only KD45 properties. Finally, $(6)\Rightarrow (1)$ holds in any class of frames contained in KD4: In fact, to show the contraposition, suppose that $\varphi$ is believable in KD4. Then, $M,w\models\Box\varphi$ for some KD4 model $M,w$, so by seriality there is a $v$ such that $wRv$ and $M,v\models\varphi$. By transitivity, we also have $M,v\models\Box\varphi$. Let $M_v$ be the submodel of $M$ generated by $v$. Then, we have $M_v,v\models\varphi$ and $M_v,v\models\Box\varphi$ hence $M_v=M_v|_{\varphi}$. Thus, $\varphi$ is not eventually self-refuting.  

In summary, we have $(1)\Leftrightarrow (2)\Leftrightarrow (3)\Rightarrow (5)\Leftarrow (6)\Rightarrow (1)$ in KD4 and S4 and $(1)\Leftrightarrow (2)\Leftrightarrow (3)\Leftrightarrow (4)\Leftrightarrow (5)\Leftarrow (6)\Rightarrow (1)$ in KD45. The remaining directions we leave to future work (some of them are expected to fail).
\end{remark}
In KD45 and S5, the formula $\varphi$ in the set $B_k\cup\{\varphi\}\cup\Gamma_k$ in Theorem~\ref{Unknowability in K and KD} simplifies to $\bigvee_{k\in S}\delta_k^\alpha$ thanks to introspections, allowing us to consider objective (purely propositional) facts only. Also, asserting $p\land\lnot\Box p$ ($p$ but I don't believe $p$) and hearing the announcement of $p\land\lnot\Box p$ ($p$ but you don't know $p$) only differ in the subjects (I/you), so it indeed makes sense that the static notions and the dynamic notions in Theorem~\ref{Unknowability and unbelievability in KD45 and S5} are equivalent.
\begin{example}
    The Moore sentence $\varphi:=p\land\lnot\Box p$ is unbelievable in KD45 and S5. In fact, it is equivalent to $p\land\Diamond\lnot p$, and for $B=\varnothing$ and $\Gamma=\{\lnot p\}$, $B\cup \{p\}\cup \{\lnot p\}$ is unsatisfiable. However, the disjunction of the two Moore sentences $\psi:=(p\land\lnot \Box p)\lor (q\land\lnot\Box q)$ is believable in KD45 and S5. In fact, for $B=\varnothing$ and $\Gamma=\{\lnot p, \lnot q\}$, $B\cup\{p\lor q\}$ is satisfiable and both $B\cup\{p\lor q\}\cup \{\lnot p\}$ and $B\cup\{p\lor q\}\cup\{\lnot q\}$ are satisfiable. Thus, $\psi$ is believable.
\end{example}
\section{Unknowability and unbelievability in the multi-agent case}
Next, we consider unknowable and unbelievable formulas in multi-agent settings.
\begin{definition}Define \emph{formulas} in multi-agent epistemic logic by
    \begin{equation*}
        \varphi:=p\mid \lnot\varphi\mid \varphi\land\psi\mid \Box_i\varphi\quad (p\in\mathbf{Prop} \text{ and } i\in G)
    \end{equation*}
\end{definition}
\begin{definition}
    Let $\varphi$ be a formula, $i\in G$ be an agent, $M$ be a model, and $\mathcal{K}$ be a class of frames.
\begin{itemize}
    \item $\varphi$ is \emph{believed} by $i$ at $w$ in $M$ iff $M,w\models\Box_i\varphi$
    \item $\varphi$ is \emph{believable} to $i$ in $\mathcal{K}$ iff $\Box_i\varphi$ is satisfiable in $\mathcal{K}$
    \item $\varphi$ is \emph{known} to $i$ at $w$ in $M$ iff $M,w\models\varphi\land\Box_i\varphi$
    \item $\varphi$ is \emph{knowable} to $i$ in $\mathcal{K}$ iff $\varphi\land\Box_i\varphi$ is satisfiable in $\mathcal{K}$
\end{itemize}
``\emph{Unbelievable}'' means ``not believable'' and ``\emph{unknowable}'' means ``not knowable.''
\end{definition}
Next, we introduce the notions of $i$-objective formulas, $i$-independent formulas, and $i$-DNF, which are needed for the multi-agent case.
\begin{definition}
    A formula is \emph{$i$-objective} iff all of its top modal operators are neither $\Box_i$ nor $\Diamond_i$. A formula is $i$-\emph{independent} iff it contains neither $\Box_i$ nor $\Diamond_i$. 
    A formula is in \emph{disjunctive normal form} for $i$ (\emph{$i$-DNF}) in K45 iff it is a disjunction of conjunctions of the form $\delta\equiv\alpha\land\Box_i\beta_1\land\cdots\land\Box_i\beta_n\land\Diamond_i\gamma_1\land\cdots\land\Diamond_i\gamma_m$ where $\alpha$ and each $\gamma_i$ are conjunctions of $i$-objective formulas and each $\beta_i$ is a disjunction of $i$-objective formulas.
\end{definition}
For example, $\varphi:=\Box_1\Diamond_2 p\land \Diamond_3 q$ is 2-objective since the top modal operators are $\Box_1$ and $\Diamond_3$, not $\Box_2$ or $\Diamond_2$. $\varphi$ is not $i$-independent for all $i\in \{1,2,3\}$.

The multi-agent version of Lemmas~\ref{modal agreement lemma} and \ref{Normal form of Boxvarphi in KD45} also hold.
\begin{lemma}\label{modal agreement lemma multi-agent}
    Let $M=(W,\{R_i\}_{i\in G},V)$ be a KD45 (or S5) model. If $wR_iv$, then for any formula $\varphi$ of the form $\Box_i\psi$ or $\Diamond_i\psi$, we have 
    \begin{equation*}
        M,w\models\varphi\Longleftrightarrow M,v\models\varphi
    \end{equation*}
\end{lemma}
\begin{lemma}\label{DNF of Box Multi-agent}
    Let $i \in G$ be any agent and $\bigvee_k \delta_k$ be a disjunctive normal form of $\varphi$ for $i$. Then in KD45 and S5, $\Box_i\varphi$ is logically equivalent to
    \begin{equation*}
        \bigvee_{S\subseteq K}\left[\bigwedge_{k\in S}\delta_k^{\Box_i\Diamond_i}\land\Box_i(\bigvee_{k\in S}\delta_k^{\alpha})\right]
    \end{equation*}
\end{lemma}
\begin{lemma}\label{multi-agent satisfiability lem}
     Let $\beta$ be an $i$-independent formula and $\Gamma$ be a finite set of $i$-independent formulas. Then, in KD45 and S5, $\Box_i\beta\land \bigwedge_{\gamma\in\Gamma}\Diamond_i\gamma$ is satisfiable iff both of the following conditions hold:
    \begin{enumerate}
        \item $\beta$ is satisfiable.
        \item For every $\gamma\in\Gamma$, $\beta\land\gamma$ is satisfiable.
    \end{enumerate}
\end{lemma}
\begin{proof}
$(\Rightarrow)$ is the same as Lemma~\ref{satisfiability and modal depth zero}. $(\Leftarrow)$ Suppose (1) and (2). Then, there exists a model $M,s$ such that $M,s\models \beta$, and for every $\gamma\in \Gamma$ there exist models $M_\gamma, s_\gamma$ such that $M_\gamma, s_\gamma\models\beta\land\gamma$ (we assume that these models are disjoint). Let $M_{\text{all}}=(W, \{R_j\}_{j\in G}, V)$ be the disjoint union of all of these models. Let $S=\{s_\beta\}\cup \{s_\gamma\}_{\gamma\in\Gamma}$. We define the rooted KD45 model $M^*=(W^*, \{R_j^*\}_{j\in G}, V^*)$ as follows. First, let $W^*=W\cup\{w^*\}$ and $V^*$ be an extension of $V$. Define the relations by $R_i^*=\{(w^*,s)\colon s\in S\}\cup (S\times S)\cup \{(x,x)\colon x\in W\backslash S\}$ and $R_j^*=R_j\cup\{(w^*,w^*)\}\, (j\neq i)$. Then, since the truth value of an $i$-independent formula does not depend on the structure of $R_i$, we have $M^*, w^*\models\Box_i\beta\land \bigwedge_{\gamma\in\Gamma}\Diamond_i\gamma$.
\end{proof}
\begin{lemma}
    In KD45 and S5, every formula is logically equivalent to a formula in disjunctive normal form for $i$ in K45.
\end{lemma}
\begin{proof}
    We induct on formulas $\varphi$. If $\varphi$ is a propositional formula, we are done. If $\varphi=\psi_1\lor\psi_2$, then $\psi_1$ and $\psi_2$ are equivalent to some $i$-DNFs by inductive hypothesis (IH), and their disjunctions are also in $i$-DNF. If $\varphi=\psi_1\land\psi_2$, then by IH, $\psi_1$, $\psi_2$ are equivalent to some $i$-DNFs $\bigvee_k\delta_k$, $\bigvee_l \delta_l'$, respectively. Thus, $\varphi\equiv(\bigvee_k\delta_k\land\bigvee_l \delta_l')\equiv \bigvee_{k,l}(\delta_k\land\delta_l')$. Since $\delta_k=\delta_k^\alpha\land\delta_k^{\Box_i\Diamond_i}$ and $\delta_l'=\delta_l'^\alpha\land\delta_l'^{\Box_i\Diamond_i}$, we have $\varphi\equiv\bigvee_{k,l}[(\delta_k^\alpha\land\delta_l'^\alpha)\land(\delta_k^{\Box_i\Diamond_i}\land \delta_l'^{\Box_i\Diamond_i})]$. This is in $i$-DNF. If $\varphi=\Box_j\psi$ for some $j\neq i$, we are done. If $\varphi=\Diamond_i\psi$, $\psi$ is equivalent to some $i$-DNF $\bigvee_k \delta_k$. Then, we have $\varphi\equiv\Diamond_i\bigvee_k(\delta_k^\alpha\land \delta_k^{\Box_i\Diamond_i})\equiv\bigvee_k\Diamond_i(\delta_k^\alpha\land \delta_k^{\Box_i\Diamond_i})\equiv\bigvee_k(\Diamond_i\delta_k^\alpha\land \delta_k^{\Box_i\Diamond_i})$. Here, we applied the reduction rules for K45. The rightmost hand becomes $i$-DNF when ordered properly. If $\varphi=\Box_i\psi$, $\psi$ is equivalent to some $i$-DNF $\bigvee_k\delta_k$ by IH. Then, by Lemma~\ref{DNF of Box Multi-agent}, $\varphi\equiv\bigvee_{S\subseteq K}\left[\bigwedge_{k\in S}\delta_k^{\Box_i\Diamond_i}\land\Box_i(\bigvee_{k\in S}\delta_k^{\alpha})\right]$, which is in $i$-DNF.
\end{proof}
Now, we state the multi-agent version of Theorem~\ref{Unknowability and unbelievability in KD45 and S5} for the static notions only. Unlike the single agent case, not all formulas in $B_S$ and $\Gamma_S$ can be purely propositional due to their interactive nature: for example in KD45, $\Box\Box$ collapses to $\Box$ in the single agent case while $\Box_1\Box_2$ does not. However, they can be $i$-objective, meaning that in the multi-agent case, Moorean phenomena occur through objective facts and beliefs of other agents.
\begin{theorem}[Unknowability and unbelievability in KD45 and S5 in the multi-agent case]\label{Multi-agent unknowability}
    Let $\bigvee_{k\in K}\delta_k$ be a disjunctive normal form of $\varphi$ for agent $i$ in K45 and for each $S\subseteq K$, let $B_S$, $\Gamma_S$ be the sets of formulas $\beta$, $\gamma$ such that $\Box_i\beta$, $\Diamond_i\gamma$ are in the formula $\bigwedge_{k\in S}\delta_k^{\Box_i\Diamond_i}$, respectively. Then in KD45 and S5, the following are equivalent for all $i\in G$:
    \begin{enumerate}
        \item $\varphi$ is unbelievable to $i$
        \item $\varphi$ is unknowable to $i$
        \item $\varphi\leftrightarrow (\varphi\land\lnot\Box_i\varphi)$ is valid.
        \item $\bigvee_{S\subseteq K}\left[\bigwedge_{k\in S}\delta_k^{\Box_i\Diamond_i}\land\Box_i(\bigvee_{k\in S}\delta_k^{\alpha})\right]$ is unsatisfiable.
    \end{enumerate}
        In particular, if the formula $\bigvee_{k\in S}\delta_k^\alpha$ and all the formulas in $B_S$ and $\Gamma_S$ are $i$-independent for all $S\subseteq K$, we can simplify condition (4) to:
    \begin{enumerate}
        \item[4'.] For any $S\subseteq K$, (i) $B_S\cup\{\bigvee_{k\in S}\delta_k^{\alpha}\}$ is unsatisfiable or (ii) there exists a $\gamma\in \Gamma_S$ such that $B_S\cup\{\bigvee_{k\in S}\delta_k^{\alpha}\}\cup\{\gamma\}$ is unsatisfiable.
    \end{enumerate}
\end{theorem}
\begin{proof}
    The proofs are essentially the same as Theorem~\ref{Unknowability and unbelievability in KD45 and S5}. The equivalence $(4)\Leftrightarrow (4')$ in the ``In particular'' part follows from Lemma~\ref{multi-agent satisfiability lem}.
\end{proof}
The following equivalence is expected to hold in S5 for multi-agent epistemic logic with common knowledge operator $C_G$, with almost the same proof in Proposition 5.7 in~\cite{holliday2010}. We just leave it for future work: $\varphi$ is always informative $\iff$ $C_G\varphi$ is unsatisfiable $\iff$ $\varphi$ is eventually self-refuting.

Several new mechanisms for Moorean phenomena arise in the multi-agent case.
\begin{example}
    $\varphi:=\Box_1 p\land\Box_2 \lnot p$ is unbelievable to both 1 and 2 in $S5$, simply because it is unsatisfiable (agent 1 knows $p$ while agent 2 knows $\lnot p$). $\psi:=\Box_1\Diamond_2 p\land\lnot p$ is unbelievable to 2 in S5 since it reduces to the Moore sentence $\Diamond_2 p\land\lnot p$ due to axiom T. However, $\psi$ is clearly believable to 2 in KD45. $\chi:=\Box_1 p\land \Diamond_2\Diamond_1\lnot p$ is unbelievable to 2 in KD45 since it is equivalent to the generalized Moore sentence $(\Box_1 p)\land\Diamond_2 \lnot (\Box_1 p)$ for 2. It is easy to see that $\psi$ and $\chi$ are believable to 1 in KD45.
\end{example}
\section{Is the BK paradox a Moorean phenomenon?}
Before concluding, we briefly analyze the Brandenburger-Keisler paradox in epistemic game theory as an example of an unbelievable formula. The paradox arises from considering the following sentence, with the understanding that ``assumption'' means ``strongest belief'':
\begin{quote}Ann believes that Bob assumes that Ann believes that Bob's assumption is wrong.
\end{quote}
According to \cite{brandenburger2006}, this sentence is a Liar-style paradox since it involves the semantic notion of ``wrong'' while the formal version expressed by modal logic is a Russell-style paradox.  In this section, we briefly introduce the BK paradox and discuss whether it can be considered a Moorean phenomenon. The notion of assumption appears in discussions of the relationship between common assumption of rationality and elimination of weakly dominated strategies in game theory. For details of the paradox, see the original paper \cite{brandenburger2006}.

\begin{definition}
    Define our formulas by
    \begin{equation*}
        \varphi:=p\mid \lnot\varphi\mid\varphi\land\psi\mid\Box_i\varphi\mid\boxplus_i\varphi\quad (p\in\mathbf{Prop} \text{ and } i\in G)
    \end{equation*}
\end{definition}
We define the truth of the assumption operator as:
\begin{itemize}
\item $M,w\models\boxplus_i\varphi$ iff for all $v\in W$, $wR_iv\Longleftrightarrow M,v\models\varphi$.
\end{itemize}
    The two-person BK paradox is formalized as in Proposition~\ref{two person BK paradox}. Here, the informal statement ``at state $w$, Ann believes that Bob's assumption (about $w$) is wrong'' is formalized as $M,w\models D$. This is because $M,w\models D$ implies $\forall v\,[w R_a v\rightarrow \lnot vR_bw]$, meaning that at any state $v$ Ann considers possible from $w$, Bob considers $w$ impossible (even though Ann is actually at $w$).\footnote{The property $\forall v\,[w R_a v\rightarrow \lnot vR_bw]$ is not expressible in the usual modal logic. However, it can be expressed in hybrid logic, as suggested in \cite{pacuit2007}.}
\begin{proposition}[The BK paradox]\label{two person BK paradox}
    Let $M=(W, R^a, R^b, V)$ be a model such that $R_a$, $R_b$ are serial and $V(D)=\{w\in W\colon\forall v\,(wR_a v\to \lnot vR_bw)\}$. Then there is no state $w\in W$ such that $M,w\models \Box_a\boxplus_b D$.
\end{proposition}
We state in our language the multi-agent generalization of the paradox in \cite{abramsky2012}.
\begin{proposition}[The multi-agent version of the BK paradox]\label{Multi-agent BK paradox}
        Let $M=(W, \{R_i\}_{i=1}^n, V)$ be a model such that each $R_i$ is serial and $V(D)=\{w_1\in W\colon\forall w_2\cdots \forall w_n\,[(w_1R_1w_2\land\cdots \land w_{n-1}R_{n-1}w_n)\to\lnot w_nR_nw_1]\}$. Then, there is no $w_1\in W$ such that 
    \begin{equation*}
        M,w_1\models \Box_1\cdots\Box_{n-1}\boxplus_n D
    \end{equation*}
\end{proposition}
\begin{proof}
To derive a contradiction, suppose that there is an $w_1\in W$ which satisfies the formula above. 
We ask whether $M,w_1 \models D$.
Suppose $M,w_1 \models D$. By seriality, there are $w_2,\dots,w_n \in W$ such that $w_1 R_1 w_2 \wedge \dots \wedge w_{n-1} R_{n-1} w_n$. Since $M,w_1 \models D$, we have $\neg w_n R_n w_1$. But by $M,w_1 \models \Box_1 \dots \Box_{n-1} \boxplus_n D$, we have $w_n R_n w_1 \iff M,w_1 \models D$, a contradiction.
Suppose $M,w_1 \models \neg D$. Then there are $w_2, \dots, w_n \in W$ such that $w_1 R_1 w_2 \wedge \dots \wedge w_{n-1} R_{n-1} w_n$ and $w_n R_n w_1$. By $M,w_1 \models \Box_1 \dots \Box_{n-1} \boxplus_n D$, we have $w_n R_n w_1 \iff M,w_1 \models D$, a contradiction.
\end{proof}
Now, let us return to the question of whether the BK paradox is considered a Moorean phenomenon. Since $\varphi:=\boxplus_b D$ also satisfies $\varphi\leftrightarrow(\varphi\land\lnot\Box_a\varphi)$, $\varphi$ still manifests a Moorean phenomenon for Ann in the \emph{generalized sense} discussed after Theorem~\ref{Unknowability in K and KD}. However, $\varphi$ is already unbelievable to Ann due to a mechanism distinct from introspections, which are required for the unbelievability of the Moore sentence. In fact, as we saw in the proof of Proposition~\ref{Multi-agent BK paradox}, the unbelievability of $\varphi$ stems from diagonalization akin to Russell's paradox. This discussion also applies to the $n=1$ case with $D=\{w\in W\colon \forall v\,[wRv\to\lnot vRw]\}$, although $\boxplus D$ is already unsatisfiable in KD.
\section{Conclusion and future directions}
In this work, we proposed static definitions of unknowability and unbelievability. We demonstrated that Moorean phenomena are best captured by contradictions among propositions that an agent already believed/considered possible, and the proposition he is now trying to believe. Our main result showed that in S5, the static notions and the dynamic notions are all equivalent to Moorean phenomena, and we further extended it to the multi-agent case for the static notions. We also analyzed the BK paradox as an example of a formula that is unbelievable due to a distinct mechanism although it is Moorean in the generalized sense.

Possible future directions are as follows. First, there might be formulas other than the BK paradox that are already unbelievable in KD or KD5. Finding such formulas would reveal further distinct mechanisms for unknowability and unbelievability. Second, we leave the discussions of whether the equivalence in Theorems~\ref{Unknowability and unbelievability in KD45 and S5} and \ref{Multi-agent unknowability} holds in other classes of frames as a future work. A generalization of Theorem~\ref{Multi-agent unknowability} for the dynamic notions and adding common knowledge and distributed knowledge operators is also left. Finally, incorporating awareness might yield new mechanisms.       
\begin{credits}
\subsubsection{\ackname}
I would like to thank Google DeepMind's Gemini 2.5 for discussions. I also thank Koki Okura, who kindly listened to my research progress in a weekly seminar. 
\subsubsection{\discintname}
The author has no competing interests to declare that are
relevant to the content of this article.
\end{credits}

\end{document}